\documentclass{article}
\usepackage[letterpaper, portrait, margin=1in]{geometry}
\usepackage[utf8]{inputenc}
\usepackage{graphicx, caption}
\usepackage{amsfonts}
\usepackage{amsmath,bm}
\usepackage{cite}
\usepackage{IEEEtrantools}
\usepackage{amssymb}
\usepackage{amsthm}
\usepackage{color}
\usepackage{comment}
\usepackage{mathtools}
\usepackage{braket}
\usepackage{enumerate}
\usepackage{multirow}
\usepackage{hyperref}
\usepackage{balance}
\usepackage{algorithm,algpseudocode}
\usepackage{flushend}
\usepackage{algorithm}
\usepackage{algpseudocode}
\usepackage{tikz}
\usetikzlibrary{quantikz}
\usepackage{authblk}
\usepackage{subcaption}

\newcommand{\CC}{\mathbb{C}}

\newcommand{\E}{\mathbb{E}}
\newcommand{\F}{\mathbb{F}}
\newcommand{\tr}{\textnormal{tr}}

\newcommand{\R}{\mathbb{R}}

\theoremstyle{plain}
\newtheorem{thm}{Theorem}

\newtheorem{defn}[thm]{Definition} 
\newtheorem{fact}[thm]{Fact} 
\newtheorem{example}[thm]{Example} 

\newtheorem{proposition}[thm]{Proposition}
\newtheorem{lemma}[thm]{Lemma}
\newtheorem{corollary}[thm]{Corollary}

\newtheorem{remark}[thm]{Remark}

\newcommand\blfootnote[1]{%
  \begingroup
  \renewcommand\thefootnote{}%
  \footnotetext{#1}%
  \endgroup
}

\title{Quasi-optimal quantum Markov chain spectral gap  estimation}
\author[1]{Adam Connolly}
\author[1,2]{Steven Herbert}
\author[3]{Julien Sorci}
\affil[1]{Quantinuum, Cambridge, UK}
\affil[2]{Department of Computer Science and Technology, University of Cambridge, UK}
\affil[3]{Quantinuum, Arlington, VA, USA}
\date{\today} 

\begin{document}

\maketitle

\begin{abstract}
This paper proposes a quantum algorithm for Markov chain spectral gap estimation that is quasi-optimal (i.e., optimal up to a polylogarithmic factor) in the number of vertices for all parameters, and additionally quasi-optimal in the reciprocal of the spectral gap itself, if the permitted relative error is above some critical value. In particular, these results constitute an almost quadratic advantage over the best-possible classical algorithm. Our algorithm also improves on the quantum state of the art, and we contend that this is not just theoretically interesting but also potentially practically impactful in real-world applications: knowing a Markov chain's spectral gap can speed-up sampling in Markov chain Monte Carlo. 

Our approach uses the quantum singular value transformation, and as a result we also develop some theory around block-encoding Markov chain transition matrices, which is potentially of independent interest. In particular, we introduce explicit block-encoding methods for the transition matrices of two algebraically-defined classes of Markov chains.
\end{abstract}

\section{Introduction}
\label{sec:intro}

A\blfootnote{Contact: \texttt{Firstname.Lastname AT Quantinuum.com}. Note that author order is alphabetical.} modern and exciting approach to quantum algorithm design is to view quantum computers as machines that manipulate the singular values and eigenvalues of large matrices. This is the essence of the celebrated \textit{quantum singular value transformation} (QSVT) \cite{Gily2019}, lauded as the \textit{grand unification of quantum algorithms} \cite{Grand-Unification}. As a result, the search for applications which reduce to singular value and eigenvalue transformations has taken centre-stage in contemporary quantum algorithm research and development. It follows that Markov chain mixing would ostensibly appear to be an excellent candidate for quantum speed-up by the QSVT, as the process of mixing is nothing more than the suppression of all but the largest eigenvalue of the Markov chain's \textit{transition matrix}.

The reality is, however, a little more complicated. The goal of mixing is to obtain the \textit{stationary distribution} and Markov chain Monte Carlo (MCMC) algorithms (which are used in an abundance of real-world applications, such as image restoration \cite{MCMC1}, genetics and bio-informatics \cite{MCMC2, MCMC3}, ecology \cite{MCMC4, MCMC5}, statistical physics \cite{MCMC6}, econometrics \cite{MCMC7} and of course machine learning \cite{MCMC8, MCMC9}, amongst many more) leverage this property when the distribution in question is not easy to sample by other means. The critical problem is that the QSVT only naturally applies to Markov chains whose transition matrices are normal -- a condition that implies that the stationary distribution is uniform (see Lemma~\ref{lem:normal-is-uniform}, below). As the uniform distribution is not hard to obtain by other means, this renders the whole exercise somewhat pointless.

In this paper we contend that quantum computing offers an important advantage in MCMC not by speeding up mixing itself, but rather by speeding-up the \textit{estimation} of the mixing time, that is, the number of Markov chains steps required to obtain the stationary distribution. This is especially well-motivated in the case of reversible Markov chains (defined below in Def.~\ref{def:reversible-MC}), as are the Markov chains deployed in the ubiquitous \textit{Metropolis-Hastings} MCMC algorithm \cite{Metropolis1953, Hastings1970}, as the estimation of the mixing time is a much harder problem than mixing itself (see Section~\ref{subsec:estimate-relax-time} for details). In particular, even though the mixing cannot directly be fast-forwarded quantumly (except when the stationary distribution is uniform), using a quantum algorithm to more accurately estimate a lower bound on the mixing time can reduce the overall cost of sampling by reducing the number of Markov chain steps needed for each sample. Our work also encompasses the case where the stationary distribution \textit{is} uniform. In such cases, one may still be interested in how long it takes to obtain the stationary distribution, a well-known example being the work of Bayer and Diaconis showing that a 52 card pack requires only seven riffle shuffles to be well-mixed \cite{BayerDiaconis1992}. Diaconis later wrote an excellent survey of the field in general \cite{Diaconis2005}.

In this paper we propose a quantum algorithm that estimates the second largest singular value of a block-encoded matrix, and hence the \textit{singular gap} (the difference between the largest two singular values), when the largest singular value is known (in our application it is known to be one). This can be used to estimate the spectral gap of reversible Markov chains by taking advantage of a known method of block-encoding the Markov chain's symmetrised discriminant \cite{Szegedy}, and using the fact that the singular gap of the symmetrised discriminant coincides with the spectral gap of the transition matrix for reversible Markov chains. The reciprocal of the spectral gap, known as the Markov chain's \textit{relaxation time}, is proportional to widely-used upper- and lower-bounds on the mixing time (see Fact~\ref{fact1}, below), and so is a useful proxy for the mixing time. As well as reversible Markov chains, our bounds also apply to Markov chains whose stationary distribution is uniform, as in this case the reciprocal of the singular gap of the transition matrix is a proxy for the mixing time. We show that, in principle, unscaled block-encoding of transition matrices of doubly-stochastic Markov chains is always possible, and give some explicit methods for certain algebraically-defined Markov chains. Our quantum algorithm has two variants, the first quasi-optimal in the Markov chain size (number of vertices) and the second quasi-optimal in the reciprocal of the singular gap itself.

\subsection*{Main contributions}

\begin{enumerate}[(i)]
    \item \textbf{Block-encoding transition matrices.} We give a simple result about when a transition matrix can be block-encoded without scaling. We also propose explicit block-encoding methods for transition matrices of certain algebraically-defined Markov chains.
    \item \textbf{Design of a new QSVT polynomial for this application.} Our proposed algorithm involves applying a filter to the singular values of a block-encoded transition matrix. This can be done quasi-optimally in the Markov chain size (number of vertices) using the polynomial approximation of the sign function as also used in the original QSVT paper for singular value / eigenvalue filtering \cite[Lemma 14]{Gily2019}, however quasi-optimality in the reciprocal of the spectral gap required us to design a new singular value filter by tuning the Dolph-Chebyshev window \cite{Dolph1946} accordingly.
    \item \textbf{Constructing a suitable ensemble of starting states.} Any technique of the form we propose for this purpose relies on a certain uniformity in the overlap of the initial states and the eigenvectors of the Markov chain. We explicitly show that unitary 2-designs suffice for this purpose.
    \item \textbf{Proof of quasi-optimality}.
\end{enumerate}

Note that ``quasi-optimality'' is a widely-used term, but it is worth defining:

\begin{defn}
    An algorithm is called \textbf{quasi-optimal} if it is only worse than the optimal algorithm for the problem by only a poly-logarithmic factor.
\end{defn}

\section{Preliminaries on Markov chains}
\label{sec:prelims}

Markov chains are characterised by the ``Markov'' property of being \textit{memoryless}: their future evolution depends only on the present state, and not the entire history of how that state was arrived at. In this article we consider only Markov chains that have discrete state-space, and evolve (make transitions) by taking ``steps'' at discrete time intervals. As such, the behaviour of a Markov chain on some state-space $\Omega$ is captured by an $N \times N$ real matrix (where $N = |\Omega|$), known as its \textit{transition matrix}, $P$, which is such that:
\begin{equation}
    \boldsymbol{\mu}(t+1) = \boldsymbol{\mu}(t) P
\end{equation}
where $\boldsymbol{\mu}(t)$ is an $N$-element non-negative real vector such that $\sum_{i=1}^N \mu_i(t) = 1$. Thus $\boldsymbol{\mu}(t)$ is a probability distribution on $\Omega$, parameterised by the time, $t$\footnote{In this paper we adopt the convention that the probability distribution left-multiplies the transition matrices, however in some of the literature right-multiplication is used.}. 

The transition matrix maps a probability distribution to another probability distribution, and so has the property:
\begin{equation}
\label{eq:stoch-def}
    \boldsymbol{1} = P \boldsymbol{1}
\end{equation}
and such matrices are known as \textit{row-stochastic}. Moreover, matrices where $\boldsymbol{1} = P^T \boldsymbol{1}$ also holds are known as \textit{doubly-stochastic}.

The eigenvalues of $P$ all have absolute value at most equal to one, and every Markov chain has at least one eigenvalue equal to $+1$ (a trivial implication of (\ref{eq:stoch-def})). Vectors in the $+1$-left-eigenspace are known as stationary distributions. Following practical and theoretical precedent, in this paper we mostly restrict our attention to \textit{ergodic} Markov chains, which are such that every state is reachable with non-zero probability from every other state within a finite number of steps. Ergodic Markov chains have a unique stationary distribution, denoted $\boldsymbol{\pi}$, and so,
\begin{equation}
    \boldsymbol{\pi} = \boldsymbol{\pi}P
\end{equation}

As well as a transition matrix, every Markov chain has a discriminant, $D'$, which is an $N \times N$ matrix such that:
\begin{equation}
    D'_{i,j} = \frac{\sqrt{\pi_i} P_{i,j}}{\sqrt{\pi_j}}
\end{equation}

Furthermore every Markov chain has a \textit{symmetrised} discriminant, $D$, which is an $N \times N$ matrix such that:
\begin{equation}
    D_{i,j} = \sqrt{P_{i,j} P_{j,i}}
\end{equation}
As $D$ is real and symmetric (therefore Hermitian), its spectrum is real.

\subsection{Some notable classes of Markov chains}

\begin{defn}
\label{def:reversible-MC}
A Markov chain is \textbf{reversible} if
\begin{equation}
    \pi_i P_{i,j} = \pi_j P_{j,i}
\end{equation}
for all $i,j \in \Omega$. These equations are called the \textit{detailed balance equations}.
\end{defn}

\begin{lemma}
\label{lem:sym-eq-disc}
    For reversible Markov chains, $D = D'$.
\end{lemma}
\begin{proof}
    \begin{equation}
        D'_{i,j} = \frac{\sqrt{\pi_i} P_{i,j}}{\sqrt{\pi_j}} = \sqrt{ \frac{\pi_i (P_{i,j})^2}{\pi_j}} = \sqrt{  (P_{i,j})^2 \frac{P_{j,i}}{P_{i,j}}}  = \sqrt{P_{i,j} P_{j,i}} = D_{i,j}
    \end{equation}
    where $\frac{\pi_i}{\pi_j} = \frac{P_{j,i}}{P_{i,j}}$ is a rearrangement of the detailed balance equation.
\end{proof}

Reversible Markov chains have the property that for each, the spectrum of its transition matrix coincides with that of its discriminant \cite{MagniezNayakRolandSantha2011} and therefore that of its symetrised discriminant (a trivial implication of Lemma~\ref{lem:sym-eq-disc}). As $D$ is symmetric, its singular values are equal in magnitude to its eigenvalues, and so we can conclude that for any reversible Markov chain, the singular values of $D$ are equal to the magnitudes of the eigenvalues of $P$. Moreover, as $D$ is symmetric, its spectrum -- and therefore that of $P$ -- is real.

\begin{defn} \textbf{Doubly-stochastic Markov chains} are such that $P$ is doubly-stochastic, i.e., $\boldsymbol{1} = P \boldsymbol{1}$ and $\boldsymbol{1} = P^T \boldsymbol{1}$ both hold, as above.
\end{defn}

\begin{remark}
    The stationary distribution of an ergodic Markov chain is uniform if and only if its transition matrix is doubly stochastic -- a fact that is easy to verify from the definition.
\end{remark}

\begin{defn}
A Markov chain is \textbf{symmetric} if its transition matrix, $P$, is symmetric.
\end{defn}

\begin{remark}
   For symmetric Markov chains, $P=D$.
\end{remark}

\begin{lemma}
\label{lem:intersection-lem}
    A Markov chain is reversible and doubly-stochastic if and only if it is symmetric.
\end{lemma}
\begin{proof}
    If a Markov chain is both reversible and doubly-stochastic then it is symmetric: the stationary distribution of doubly-stochastic Markov chains is always uniform, so $\forall i, j, \, \pi_i = \pi_j$; substituting this into the detailed balance equation $\forall i, j, \, P_{i,j} = P_{j,i}$ and so the transition matrix is symmetric. 
    
    Conversely, if a Markov chain with transition matrix $P$ is symmetric, then $P = P^T$ and so it is doubly stochastic, and hence its stationary distribution is uniform. Substituting $\forall i, j, \, \pi_i = \pi_j$ (by double-stochasticity) and $\forall i, j, \, P_{i,j} = P_{j,i}$ (by symmetry) the detailed balance equation certainly holds, and so the Markov chain is reversible.
\end{proof}

\begin{defn} \textbf{Normal Markov chains} are such that $P$ is a normal matrix. That is, $P P^T = P^T P$, as $P$ is always real.
\end{defn}

\begin{remark}
    Every symmetric transition matrix is normal (trivially), however there are normal transition matrices that are not symmetric. For example, consider asymmetric permutation matrices which are: stochastic and unitary, hence normal.
\end{remark}

\begin{lemma}
\label{lem:normal-is-uniform}
    Every ergodic normal Markov chain is doubly-stochastic and hence has uniform stationary distribution.
\end{lemma}
\begin{proof}
    Let $P$ be the transition matrix of an ergodic Markov chain, i.e., we have $P \boldsymbol{1} = \boldsymbol{1}$ by row-stochasticity and $P P^T = P^T P$ by normality, hence:
    \begin{equation}
        P P^T \boldsymbol{1} = P^T P \boldsymbol{1} = P^T \boldsymbol{1}
    \end{equation}
    and so $P^T \boldsymbol{1}$ is a $+1$ right-eigenvector of $P$. However, for an ergodic Markov chain there is a unique $+1$ right-eigenvector, namely $\boldsymbol{1}$, and hence $P^T \boldsymbol{1} = \boldsymbol{1}$, thus $P^T$ is row-stochastic and so $P$ is doubly-stochastic.
\end{proof}

Normal matrices are unitarily diagonalisable, and so we get the standard property:

\begin{fact}
    Normal Markov chains are such that the magnitudes of their eigenvalues coincide with their singular values.
\end{fact}

\noindent Figure~\ref{fig:Matrix-Venn} shows a simple Venn diagram of the relationship between these classes of Markov chains. 

\begin{figure} [t!]
\centering
\includegraphics[ width = 0.75\linewidth]{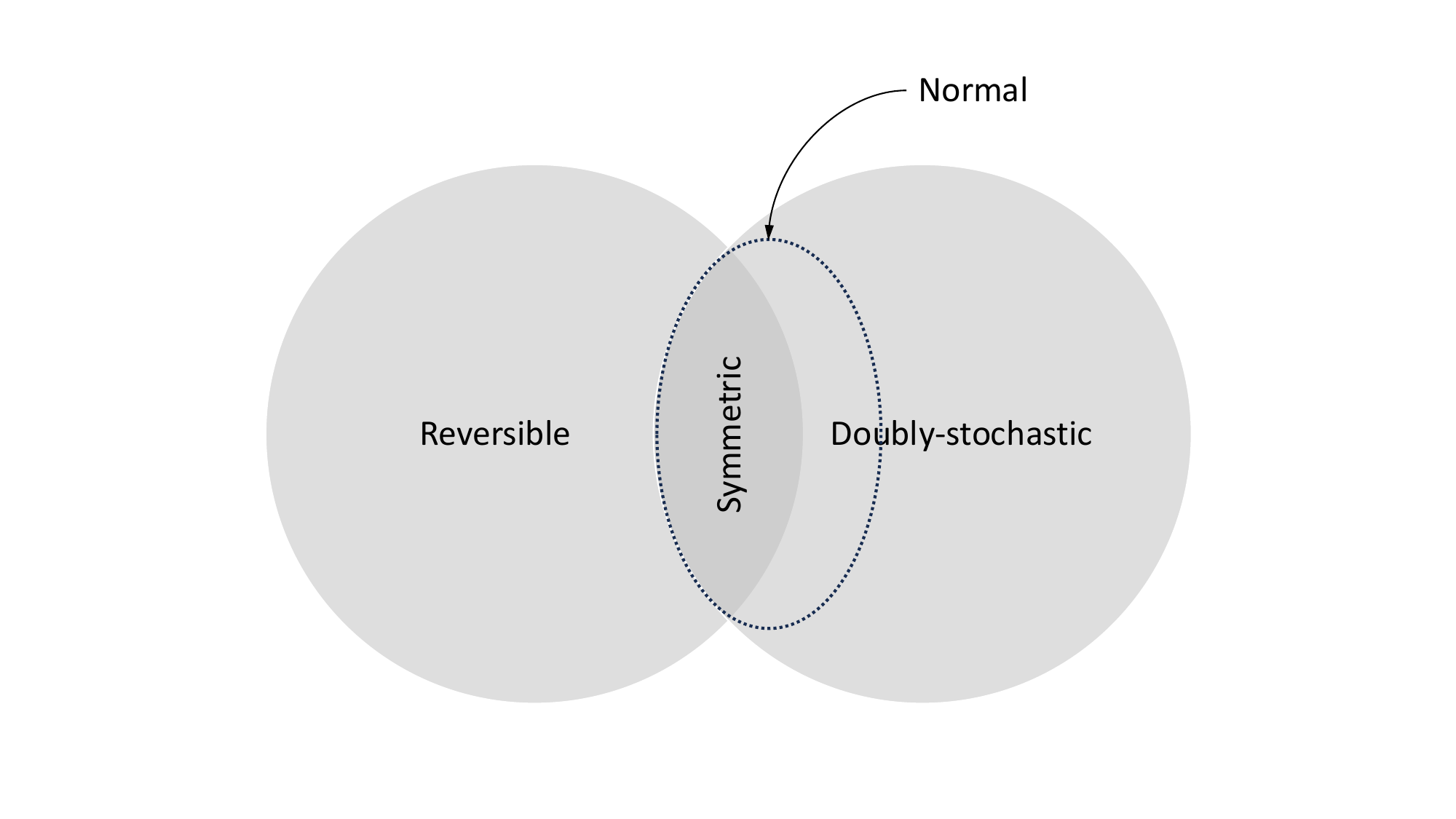}
  \captionsetup{width=.9\linewidth}
  \caption{Relationship between some important ergodic Markov chain classes.}
 \label{fig:Matrix-Venn}
\end{figure}

\subsection{Markov chain spectral gap and mixing time}

The mixing time of a Markov chain is the number of steps required to obtain the stationary distribution to some defined degree of approximation. There are various ways to measure closeness of probability distributions, which can then be used to define the mixing time, a widely-used measure for this purpose is the total variation distance, $||.||_{\text{TV}}$, which then gives 
\begin{align}
    \tau (\epsilon) & = \min_n \{ \forall x \in \Omega, \, , || P^n (x,.) - \boldsymbol{\pi}||_{\text{TV}} \leq \epsilon \} \\
\end{align}

In this paper we estimate not the mixing time itself, but rather the \textit{spectral gap}.

\begin{defn}
    Let $P$ be the transition matrix of an ergodic Markov chain with real eigenvalues $1 = \lambda_1 > \lambda_2 \geq \lambda_3 \geq \dots \geq \lambda_N \geq -1.$ \textbf{The spectral gap} of $P$ is 
    \begin{equation}
        \gamma = 1 - \max_{i \geq 2} |\lambda_i|
    \end{equation}
\end{defn}

\begin{remark}
    The spectral gap is non-negative, and strictly positive unless $\lambda_N = -1$, which occurs if and only if the graph is bipartite \cite{KontoyiannisMeyn2009}.
\end{remark}

Note that the spectral gap is defined for transition matrices with real spectra, which includes the set of reversible and normal Markov chains, as explained above. The spectral gap can be used to bound the mixing time \cite{LevinPeresWilmer2009}:

\begin{fact}
\label{fact1}
    For reversible and normal Markov chains, the following bounds hold:
    \begin{equation}
        \frac{1}{2 \gamma} \log\Big( \frac{1}{2 \epsilon}\Big) \leq \tau (\epsilon) \leq \frac{1}{\gamma} \log\Big(\frac{1}{2 \epsilon \sqrt{\pi^*}}\Big)
    \end{equation}
    where $\pi^*$ is the minimum entry of $\boldsymbol{\pi}$.
\end{fact}

As $\frac{1}{\gamma}$ is proportional to both the lower- and upper-bounds, it may be taken directly as a proxy for the mixing time and, as such, is given a name:

\begin{defn}
\label{def:relax-time}
    The \textbf{relaxation time} of a Markov chain with spectral gap $\gamma$ is $\tau_{\textnormal{rel}} = \frac{1}{\gamma}$.
\end{defn}

We are also concerned with cases where the matrix is not normal or reversible, but is still doubly-stochastic. In this case, we define:
\begin{defn}
    Let $P$ be the transition matrix of a doubly-stochastic ergodic Markov chain with singular values $1 = \sigma_1 > \sigma_2 \geq \sigma_3 \geq \dots \geq \sigma_N \geq 0.$ \textbf{The singular gap} of $P$ is 
    \begin{equation}
        \gamma_s = 1 - \sigma_2
    \end{equation}
\end{defn}
\noindent We have similar bounds on the mixing time in terms of the singular gap \cite{LevinPeresWilmer2009, SaloffCosteZuniga2009, Ganapathy2007}:
\begin{fact}
    For doubly-stochastic Markov chains, the following bounds hold:
    \begin{equation}
        \frac{1}{\gamma_s} \log\Big( \frac{1}{\epsilon}\Big) \leq \tau (\epsilon) \leq \frac{1}{\gamma_s} \log\Big(\frac{N}{\epsilon}\Big)
    \end{equation}
\end{fact}

\section{Overview of related work}

\subsection{Classical algorithms for mixing time and spectral gap estimation}
\label{subsec:estimate-relax-time}

Hsu \textit{et al} prove that the complexity of estimating the spectral gap of a reversible Markov chain is 
\begin{equation}
    \tilde{\Omega} \left( \frac{N}{\gamma} \right)
\end{equation}
with the same bound holding for estimation of the mixing time itself \cite{hsu2015mixing}. Wolfer and Kontorovich also study this question, coming up with the same bounds \cite{wolfer2022estimating}, and in a second paper provide a nice summary of the known results \cite[Table 1]{wolfer2023improved}. 

It is notable that the problem of estimating the mixing time of reversible Markov chains is as hard as mixing the chain itself in terms of the spectral gap, and much worse (linear rather than logarithmic) in terms of the size of the Markov chain. In the case of rapidly mixing Markov chains (i.e., $\tau \in \mathcal{O}(\text{poly} \log N)$), the overall complexity is therefore exponentially greater. This gives clear motivation for quantumly speeding up the estimation of the mixing time, even when the mixing itself cannot be. For example, in MCMC, the goal is to sample from the stationary distribution, and this in turn requires the chain to be mixed. Knowing how many steps this requires is critical to obtaining the stationary distribution as efficiently as possible.

\subsection{Block encoding Markov chain symmetrised discriminants}
\label{sec:encode-D}

\begin{defn}
\label{def:block-encode}
A (scaled) block encoding of a matrix $A$ is a unitary matrix $U_A$ such that $(\bra{x_1}\otimes I) U_A (\ket{x_2}\otimes I) = \frac{1}{s} A$, for some $\ket{x_1}$, $\ket{x_2}$ and real, positive constant $s$. When $s=1$ we call the block encoding unscaled, or say it is a block-encoding without scaling.
\end{defn}
\noindent In many examples $\ket{x_1} = \ket{x_2} = \ket{0}$ so that $U_A = \begin{bmatrix} \frac{1}{s} A & \cdot \\ \cdot & \cdot \end{bmatrix}$.

The core component of most of the existing work on quantising Markov chains is the Szegedy operator \cite{Szegedy}, which may be thought of as a block encoding of $D$. For reversible Markov chains, the spectrum of the transition matrix coincides with that of the symmetrised discriminant, and so this preserves some important information about the transition matrix itself. Most notably, Apers and Sarlette use the Szegedy operator to ``fast forward'' the dynamics of a Markov chain \cite{apers2019quantum}; whilst it is pointed out in the original QSVT paper how a block-encoding of $D$ can be leveraged to obtain hitting time estimates \cite{Gily2019}.

To see how $D$ may be block-encoded, for simplicity we shall assume that the Markov chain is over a graph with some $N = 2^n$ vertices; to block encode the symmetrised discriminant, first consider an operator, $U_P$ that operates on a $N^2$ element space (i.e., an $2n$ qubit operator) and encodes the transitions:
\begin{equation}
\label{eq:UP-def}
    U_P\ket{0} \ket{i} = \sum_{j=1}^N \sqrt{P_{i,j}} \ket{j} \ket{i}
\end{equation}
Note that this defines only part of the matrix, and the unspecified part can always be constructed such that $U_P$ is unitary.

Access to such an operator is essentially commensurate with classical access to look up (or compute) transition probabilities -- even when the Markov chains is too large to explicitly express the transition matrix. The $N^2$ element space can thus be thought of as a first register of $n$ qubits representing the vertex transitioned \textit{to} followed by a second register of $n$ qubits representing the vertex transitioned \textit{from}. A second operator is the SHIFT operator, denoted $\F$, which is the permutation matrix that swaps the two registers:
\begin{equation}
    \F \ket{j,i} = \ket{i,j} 
\end{equation}
A simple computation now shows that the operators $U_P$ and $\F$ are all that are required to block-encode the discriminant.
\begin{lemma}
\label{lem:UP-original}
$U_P^\dagger \, \F \, U_P$ block encodes $D$:
\end{lemma}
\begin{proof}
\begin{align}
    \bra{0} \bra{i} (U_P^\dagger \, \F \, U_P) \ket{0} \ket{j} & = \left( \sum_k \sqrt{P_{i,k}} \bra{k,i} \right) \F \left( \sum_{l} \sqrt{P_{j,l}} \ket{l,j} \right) \nonumber \\
    & = \sum_{k,l} \sqrt{P_{i,k}P_{j,l}} \braket{k,i | j,l} \nonumber \\
    & =  \sqrt{P_{i,j}P_{j,i}} = D_{i,j}
\end{align}
\end{proof}
\begin{corollary}
\label{cor:UP-original}
    If $P$ is symmetric, $U_P^\dagger \, \F \, U_P$ block encodes $P$.
\end{corollary}

\subsection{Quantum algorithms for relaxation time and spectral gap estimation}
\label{subsec:estimate-relax-time-quantum}

One of the applications suggested in Ref.~\cite{Gily2019} is the use of a singular value filter to fast-forward Markov chain mixing for a block-encoded transition matrix. Whilst, as elaborated upon in the introduction, this only naturally applies to normal Markov chains (whose stationary distribution is therefore uniform), one can appreciate that the thresholding technique in the same paper could be adapted to estimate the singular gap of a Markov chain transition matrix, with some assumptions about the overlap between the singular vectors and the initial state. This idea is exactly encompassed by the framework proposed by Zhang \textit{et al} \cite{zhang2025exponential}, who also address the need for overlap in the initial state. This is something we do in a different way, and the enumeration in the introduction summarises the other ways in which this work is differentiated from, and improves upon, existing literature.

\section{Unscaled transition matrix block encoding for certain algebraically-defined Markov chains}
\label{sec:block-encode}

The only matrices that are both unitary \textit{and} stochastic are permutation matrices, and hence permutation matrices are the only transition matrices that ``block-encode themselves'', in all other cases the unitary must be strictly larger than the block-encoded transition matrix. Usually we are concerned with the case where the matrix $A$ has to be ``scaled-down'' to fit into the unitary block encoding, hence $s > 1$ (for $s$ as in Def.~\ref{def:block-encode}). It is worth noting that standard techniques can be deployed to prepare a (scaled) block-encoding $\frac{P}{s' ||P||}$ for any $s' \geq s$ when $P$ is $s$-sparse \cite{berry2009black, Camps:2022jnx, Sunderhauf2023, qclab}.

However, moving on from generic scaled block-encoding, it is of particular interest to address the case where an unscaled block-encoding is possible in principle. This in turn, depends on the spectral norm, $||A||$, of $A$, 
\begin{fact}
\label{fact:A-block-encode}
    A matrix, $A$, has an unscaled block encoding if and only if $||A|| \leq 1$.
\end{fact}
Notably, the spectral norm is the square root of the largest eigenvalue of $A A^\dagger$, or equivalently the largest singular value of $A$. From this an important result follows:
\begin{lemma}[A special case of Ref.~{\cite[Prop 2.1]{rodtes2021structured}}]
\label{lem:when-block-encodable-a}
    For any Markov chain, 
    \begin{align}
        ||P|| & \geq 1 \\
        ||P|| & = 1 \iff P \textnormal{ is doubly-stochastic}
    \end{align}
\end{lemma}

\begin{lemma}
\label{lem:when-block-encodable}
    A stochastic matrix is block-encodable if and only if it is doubly-stochastic.
\end{lemma}
\begin{proof}
    First, we show that if $P$ is doubly stochastic then it is block-encodable: For such a doubly stochastic $P$, $P P^\dagger$ is itself a stochastic matrix. It follows that the largest eigenvalue of $P P^\dagger$ is one, and using Fact~\ref{fact:A-block-encode} this suffices to prove that every doubly-stochastic Matrix is block-encodable without scaling. That stochastic matrices that are not doubly stochastic cannot be block-encoded without scaling follows from Lemma~\ref{lem:when-block-encodable-a} and Fact~\ref{fact:A-block-encode}.
\end{proof}

Corollary~\ref{cor:UP-original} gives an explicit block encoding for one class of transition matrix, namely the symmetric stochastic matrices, and this block encoding rests on the fact that $\F$ is a permutation matrix -- as permutations are the unique class of matrices that are both stochastic and unitary (the former needed for the block-encoding to correctly represent probability values, the latter for the entire matrix $U_P^\dagger \F U_P$ to be unitary). However, $\F$ is not the only permutation matrix of dimension $N^2$, and it is interesting to note that the construction $U_P^\dagger R U_P$, where $R$ is a permutation works more generally, and we now give two classes of algebraically-defined Markov chains that can be block encoded in this manner.

\subsection{Markov chains over finite groups}
\label{subsec:group-P}

\begin{defn}
\label{def:chain-over-G}
    Let $G$ be a finite group of order $n$, and let $\mu: G \rightarrow \R$ be an arbitrary probability distribution over $G$, meaning that $\mu(x) \geq 0$ for all $x \in G$ and $\sum_{x \in G} \mu(x) = 1$. We define a Markov chain over $G$ by defining $P_{g,xg} = \mu(x)$ for all $g, x \in G$.
\end{defn}
We now prove that Markov chains over groups can be block-encoded without scaling. To begin, let 
\begin{align*}
	U_P \ket{0, g} = & \sum_{x \in G} \sqrt{P_{g, xg}} |xg, g\rangle, \\
	R|xg, g\rangle = & |x^2g, xg\rangle.
\end{align*}

Note that $U_P$ is exactly the state preparation matrix of (\ref{eq:UP-def}), just indexed in terms of group elements to aid the subsequent analysis. $U_P$ is, as above, known to be unitary, and furthermore:

\begin{lemma}
	The matrix $R$ is a permutation matrix, and is thus unitary. 
\end{lemma}

\begin{proof}
	It suffices to show that $R$ is injective on the basis states. Suppose that we have $R|xg, g\rangle = R|yh, h\rangle$ for some $g, h, x, y \in G$. Then by definition this means $|x^2g, xg\rangle = |y^2 h, yh\rangle$, which only holds if $xg = yh$. Therefore  $|x^2g, xg\rangle = |yxg, xg\rangle$. Since $x^2g = yxg$ then we must have $x = y$. Therefore $R$ is injective and therefore a permutation matrix. 
\end{proof}

From which the main result (of this subsection) follows:

\begin{thm}
	The matrix $U_P^\dagger R^\dagger U_P$ is a unitary which block encodes $P$. 
\end{thm}

\begin{proof}
	Consider the following entry
\begin{equation*}
\begin{split}
	\langle 0, h |U_P^\dagger R^\dagger U_P |0, g\rangle &= \sum_{x, y \in G} \sqrt{P_{g, xg} P_{h,yh}} \langle yh, h |R^\dagger |xg, g\rangle \\
	&= \sum_{x, y \in G} \sqrt{P_{g, xg} P_{h,yh}} \langle y^2h, yh |xg, g\rangle \\
\end{split}
\end{equation*}
Solving for the $x, y \in G$ which contribute to the sum we see this holds when $g=yh$ and $xg=y^2h$. The first equation implies that $y = gh^{-1}$, after which the second reads $xyh=y^2h$. This then implies $x = y$, so our entry is equal to 
\begin{equation*}
	\langle 0, h|U_P^\dagger R^\dagger U_P |0, g \rangle =  \sqrt{P_{yh, y^2h} P_{h,yh}} = \sqrt{\mu(y) \mu(y)} = \mu(gh^{-1}) = P_{h,g},
\end{equation*}
meaning that $U_P^\dagger R^\dagger U_P$ is indeed a block encoding of $P$.
\end{proof}

It is, of course, incumbent upon us to ask whether this is actually the construction of something \textit{new}.

\begin{lemma}
    Any Markov chain over a group, as defined in Def.~\ref{def:chain-over-G}, is reversible if and only if $\mu(x)=\mu(x^{-1})$ for all $x \in G$.
\end{lemma}
\begin{proof}
    Consider some entry $P_{g,xg} = \mu(x)$, then 
    \begin{equation}
        P^T_{g,xg} = P^T_{x^{-1}x g, xg} = P_{xg, x^{-1} (xg)} = \mu(x^{-1})
    \end{equation}
    So it follows that $\forall x, \, \mu(x) = \mu(x^{-1}) \iff P = P^T$, i.e., $P$ is symmetric, and symmetric Markov chains are always reversible. Conversely, as $P$ has an unscaled block-encoding (i.e., by the explicit construction, above), we know from Lemma~\ref{lem:when-block-encodable} that it is doubly-stochastic and if it is also reversible then it is certainly symmetric by Lemma~\ref{lem:intersection-lem}.
\end{proof}

\subsection{Markov chains on graphs with a linear ordering}

\begin{defn}
\label{def:chain-over-order}
Say $P$ is the transition matrix of a Markov chain over a set $\Omega$ endowed with a linear ordering if the underlying digraph is $d$-out-regular, and such that for all $i=1,...,d$ there is a function $f_i: \Omega \rightarrow \Omega$, where $f_i(x)$ is interpreted as the $i^{th}$ neighbor of $x$. 
\end{defn}

We define matrices $R$ and $U_P$ which act as 
\begin{equation*}
	U_P |0, x\rangle = \sum_{i=1}^n \sqrt{P_{x, f_i(x)}} | f_i(x), x\rangle
\end{equation*}
\begin{equation*}
	R|f_i(x), x\rangle = |f_i \circ f_i (x), f_i(x) \rangle
\end{equation*}
Once again note that $U_P$ is a special case of that in (\ref{eq:UP-def}), indexed to aid the following analysis. 
\begin{remark}
    $R$ acts as a permutation when, for any fixed $x \in \Omega$, $f_i(x)=f_j(x)$ implies that $i=j$, i.e., it is a simple digraph.
\end{remark}
By definition, we take a Markov chain as a simple digraph, and so $R$ is certainly unitary, as is $U_P$, hence it is enough to show that $U_P^\dagger R^\dagger U_P$ is a block encoding and unitarity is automatic.

\begin{thm}
	Suppose that $P$ is a Markov chain over $\Omega$ endowed with a linear ordering $\{f_i\}_{i=1}^d$ such that $P_{x, f_i(x)} = P_{y,f_i(y)}$ for all $x,y \in \Omega$ and $i=1,...,d$. Then $U_P^\dagger R^\dagger U_P$ is a unitary block encoding of $P$.
\end{thm}
\begin{proof}
\begin{equation*}
\begin{split}
	\langle  0, y|U_P^\dagger R^\dagger U_P | 0, x\rangle & = \sum_{i=1}^n \sum_{j=1}^n \sqrt{P_{x, f_i(x)} P_{y, f_j(y)}} \langle  f_j(y), y| R^\dagger |f_i(x), x\rangle  \\
	& = \sum_{i=1}^n \sum_{j=1}^n \sqrt{P_{x, f_i(x)} P_{y, f_j(y)}} \langle f_j \circ f_j(y), f_j(y)|  f_i(x), x\rangle  \\
\end{split}
\end{equation*}
If $x \neq f_j(y)$ for all $j$ then the sum equals $0$, which agrees with $P_{x,y}=0$. Otherwise, this picks out the term in the sum with $x=f_j(y)$, giving 
\begin{equation*}
	 \langle 0,y |U_P^\dagger R^\dagger U_P | 0, x\rangle = \sum_{i=1}^n \sqrt{P_{x, f_i (x) } P_{y, x}} \langle  f_j(x),x|f_i(x),x\rangle
\end{equation*}
Now picking out the term with $i=j$ we obtain 
\begin{equation*}
	 \langle 0, y|U_P^\dagger R^\dagger U_P |0, x\rangle = \sqrt{P_{x, f_i(x)} P_{y, x}} = \sqrt{P_{x, f_i(x)} P_{y, f_i(y)}} = P_{x, f_i(x)} = P_{x,y}
\end{equation*}
\end{proof}

Once again, the fact that it is a block encoding is enough to show that such Markov chains are certainly doubly-stochastic. Moreover, we can show that certain Markov chains that are not symmetric (and therefore are not reversible) adhere to this structure:

\begin{example}
    Consider a cycle as a Markov chain, then we have $d=1$ and $P_{x, f_1 (x)} = 1$ for the single out-edge of each vertex. So such Markov chains meet the conditions, and are thusly block-encodable, but in general are not symmetric (apart from the trivial case of a single vertex, or for non-ergodic chains where each vertex has just a self-loop).
\end{example}

\section{A quasi-optimal quantum algorithm for Markov chain spectral / singular gap estimation}

The essential idea for the singular / spectral gap estimation algorithm is to use bisection search as a loop around an algorithm that \textit{decides} if the second largest singular value is greater or smaller than some threshold (with a specified relative error). As this decision algorithm is at the heart of the overall quantum advantage, and is non-trivial, it is worth first establishing this as a claim in its own right.

\subsection{An algorithm for second largest eigenvalue thresholding}

This algorithm requires three essential components:
\begin{enumerate}
    \item An ensemble of initial states such that with high probability the overlap with any singular vector of the Markov chain transition matrix / symmetrised discriminant is approximately equal to the expectation thereof.
    \item A filter that is a polynomial function of appropriate degree to suppress the singular values below a threshold using the QSVT.
    \item Quantum counting to distinguish whether the second largest eigenvector has been filtered out or not.
\end{enumerate}

We now address these components in turn.

\subsubsection*{Preparing a suitable initial ensemble of states}

Let $N$ be the dimension of a Hilbert space, such that $n = \log_2 N$ is an integer (the number of qubits).
\begin{lemma}
    An ensemble of states that is a unitary 2-design, $\mathcal{U}$, has that
    the following properties of $\mathcal{U}$ hold for every state, $\ket{\psi}$ of dimension $n$:
\begin{align}
    \mathbb{E}_\mathcal{U} (|\bra{\psi} U \ket{0}|^2) & = \frac{1}{N} \\
    \lim_{\substack{N \to \infty}} \mathbb{E}_\mathcal{U} (|\bra{\psi} U \ket{0}|^4) - \left(\mathbb{E}_\mathcal{U} (|\bra{\psi} U \ket{0}|^2) \right)^2 & \to 0
\end{align}
\end{lemma}
\begin{remark}
\label{rem:two-design}
the first condition guarantees that the expected overlap is what we require; the second says that the variance tends to zero as $N \to \infty$, and so means that for any $0 < p < 1$ and $0 < k < 1$, there is a value $N^*$ such that for all $N > N^*$, the following holds when $U \sim \mathcal{U}$:
\begin{equation}
    \mathrm{Pr}\left(|\bra{\psi} U \ket{0}|^2 \geq \frac{1-k}{N}\right) > p
\end{equation}
and also 
\begin{equation}
    \mathrm{Pr} \left(|\bra{\psi} U \ket{0}|^2 \leq \frac{1+k}{N}\right) > p
\end{equation}
\end{remark}

\begin{proof}
    Suppose that $\mathcal{U}$ is a unitary $1$-design. By definition, this means that for any observable $O$ acting on $\CC^N$ we have
\[ \frac{1}{|\mathcal U|} \sum_{U \in \mathcal U} U O U^{\dagger} = \int_{U(N)} U O U^{\dagger} d \mu, \]
where $\mu$ denotes the Haar measure on $U(N)$. It is known that $\int_{U(N)} U O U^{\dagger} d \mu = \frac{\tr(O)}{N}I$. Taking $O =\ket{\psi} \bra{\psi}$ and considering the process where we sample a unitary uniformly at random from $\mathcal{U}$ we obtain the expectation
\begin{equation}
\begin{split}
    \E_\mathcal{U} \big[|\bra{\psi} U \ket{0}|^2\big] & = \E_\mathcal{U} \big[ \bra{0} U^\dagger \ket{\psi} \bra{\psi} U \ket{0} \big] \\
    & =\bra{0} \E_\mathcal{U} \big[  U^\dagger \ket{\psi} \bra{\psi} U  \big] \ket{0} \\
    & =\bra{0} \frac{\tr(\ket{\psi} \bra{\psi})}{N}I \ket{0} \\
    & = \frac{1}{N} \\
\end{split}
\end{equation}
This shows that a unitary $1$-design suffices to satisfy the first requirement. For the second requirement, suppose that $\mathcal{U}$ is additionally a unitary $2$-design. This means that 
\begin{equation}
    \frac{1}{|\mathcal U|} \sum_{U \in \mathcal U} U^{\otimes 2} O U^{\dagger \otimes 2} = \int_{U(N)} U^{\otimes 2} O U^{\dagger \otimes 2} d \mu,
\end{equation}
for every observable $O$. There is a well-known identity that 
\begin{equation}
    \int_{U(N)} U^{\otimes 2} O U^{\dagger \otimes 2} d \mu = \alpha I + \beta \F,
\end{equation}
where $I$ is the identity acting on $\CC^N \otimes \CC^N$, $\F$ remains the SHIFT operator acting on $\CC^N \otimes \CC^N$ as $\ket{x,y} \mapsto \ket{y, x}$, and the coefficients $\alpha, \beta$ are given by
\begin{equation}
    \alpha = \frac{\tr(O) - \frac{1}{N}\tr(\F O)}{N^2 - 1},
\end{equation}
\begin{equation}
    \beta = \frac{\tr(\F O) - \frac{1}{N}\tr(O)}{N^2 - 1}.
\end{equation}
We will use these facts along with the partial trace identity
\begin{equation}
    \tr_2 \Big((A \otimes B) \F \Big) = AB
\end{equation}
to simplify the expectation of $|\bra{\psi} U \ket{0}|^4$. Consider
\begin{equation}
\begin{split}
    |\bra{\psi} U \ket{0}|^4 &= \bra{0} \Big( U^\dagger \ket{\psi} \bra{\psi} U \Big) \Big( \ket{0} \bra{0} U^\dagger \ket{\psi} \bra{\psi} U \Big) \ket{0} \\
    &= \bra{0} \tr_2 \Big( \big( U^\dagger \ket{\psi} \bra{\psi} U \big) \otimes \big( \ket{0} \bra{0} U^\dagger \ket{\psi} \bra{\psi} U \big) \F \Big) \ket{0} \\
    &= \bra{0} \tr_2 \Big( \big(I \otimes \ket{0}\bra{0}\big) \big( U^{\dagger \otimes 2} (\ket{\psi} \bra{\psi})^{\otimes 2} U^{\otimes 2} \big) \F \Big) \ket{0} \\
\end{split}
\end{equation}
where we have used the partial trace identity above. Taking expectation and applying the identities above we obtain
\begin{equation}
\begin{split}
    \bra{0} \tr_2 \Big( \big(I \otimes \ket{0}\bra{0}\big) \big( \alpha I + \beta \F \big) \F \Big) \ket{0} &= \alpha \bra{0} \tr_2 \Big( \big(I \otimes \ket{0}\bra{0}\big) \F  \Big) \ket{0} + \beta \bra{0} \tr_2 \big( I \otimes \ket{0}\bra{0} \big) \ket{0}  \\
    &= \alpha \bra{0} \big(\ket{0}\bra{0}\big) \ket{0} + \beta \bra{0} I \ket{0}  \\
    &= \alpha + \beta  \\
\end{split}
\end{equation}
Lastly, simplifying the constants $\alpha, \beta$ in this case where $O = \ket{\psi} \bra{\psi}^{\otimes 2}$, we see that $\tr(O)=1$ and $\tr(\F O ) = \tr(O \F) = \tr(\ket{\psi} \bra{\psi} \ket{\psi} \bra{\psi}) = 1$, hence in summary we have the identity
\begin{equation}
    \E\Big[\big|\bra{\psi} U \ket{0}\big|^4\Big] = \frac{1 - \frac{1}{N}}{N^2 - 1} + \frac{1 - \frac{1}{N}}{N^2 - 1} = \frac{2}{N(N+1)}
\end{equation}
thus satisfying the second property.
\end{proof}

It is worth remarking on the fact that unitary 2-designs have been well-studied in the literature, and it is known that they can be prepared in shallow circuit depth \cite{two-design}. More generally, focusing on Haar randomness negates any ``adversarial'' attack along the lines that any proposed ensemble that does not overlap much with a certain state has the vulnerability that said state can be chosen to be the second largest eigenvector and a transition matrix constructed accordingly. However, beyond such contrived situations, it is likely that in practice other ensembles of states would fair well.

\subsubsection*{Polynomial functions for singular value filtering}

A first option is to use the polynomial transformation for singular value (or eigenvalue) filtering from Ref.~\cite[Theorem~31]{Gily2019}.

\begin{lemma}
\label{lem:alg-filter-sign}
    For every $\epsilon >0$, $0 < \Delta \leq 1$, there exists a family of polynomial functions $f_d^{(1)}(x)$ (one polynomial for each degree, $d$) such that for any $0 < t < 1$, $-1 \leq x \leq 1, \, |f_d^{(1)}(x)| \leq 1$ and
    \begin{align}
          0 \leq x \leq 1 - \Delta(1 + t) : & \,\,\,\,\,\,\, |f_d^{(1)}| < \alpha \\
        1 - \Delta(1 - t)  \leq x \leq 1 :  & \,\,\,\,\,\,\, \frac{1}{4} \leq f_d^{(1)}
    \end{align}
    Moreover, $d = \mathcal{O} \left( \frac{1}{\Delta t} \log \frac{1}{\alpha} \right)$.
\end{lemma}

If the thresholded value, $\Delta$, is to be any value between 0 and 1, this is essentially the best that we can do \cite{Gily2019, EremenkoYuditskii2010}. However, the focus in the theory of Markov chains is often when $\Delta$ is small, and in particular one is often concerned about the scaling with $\Delta$ for fixed $t$. We first consider the very special case where $t = 1$. 
\begin{lemma}
    There exists a family of polynomial functions, $f^{(2)}_d$ (one polynomial for each degree, $d$) such that
    \begin{align}
\label{eqn-filter-special-case}
        0 \leq x \leq 1-2\Delta  : & \,\,\,\,\,\,\, |f_d^{(2)}(x)| < \alpha \\
        \label{eq:f2-b}
        1-2\Delta \leq x \leq 1: & \,\,\,\,\,\,\, |f_d^{(2)}(x)| \leq 1 \\
        \label{eq:f2-c}
        x = 1 : & \,\,\,\,\,\,\, \frac{1}{4} < f_d^{(2)}(x) \leq 1
\end{align}
Moreover, $d = \mathcal{O} \left( \frac{1}{\sqrt{\Delta}} \log \frac{1}{\alpha} \right)$.
\end{lemma}
\begin{proof}
    In this case, it is well known that the monomial $x^{\tilde{d}}$ for $\tilde{d} = \mathcal{O}((2\Delta)^{-1}\log(1/\alpha))$ can achieve this. Furthermore, this monomial can be approximated to arbitrary accuracy $\tilde{\epsilon} >0$ in the range $[-1,1]$ by a polynomial $p_{d,\tilde{d}}$ of degree $d$ where $d = \mathcal{O}(\sqrt{2\tilde{d} \log(1/\tilde{\epsilon})})$. This result is a classic application of the theory of Chebychev polynomials which is described for example by Sachdeva and Vishnoi~\cite[Theorem 3.2]{sachdeva2013}. Choosing $\tilde{d}$ and $\tilde{\epsilon}$ appropriately, this results in a polynomial $f$ which satisfies all of the conditions in (\ref{eqn-filter-special-case}) -- \ref{eq:f2-c}), and has degree $\mathcal{O}(\Delta^{-1/2} \log (1/\alpha))$.  
\end{proof}

The crucial property that allows the variation with $\frac{1}{\Delta}$ to be quadratically better in the case of $f^{(2)}(x)$ relative to $f^{(1)}(x)$ is that in the former the function is growing rapidly at $x=1$, as there is no lower-bound on the value that the function must take for any $x < 1$. However, this is a little \textit{too} restrictive, and would lead to an algorithm for singular value lower-bounding, rather than thresholding. Instead, we propose the following function which utilises the fact that $\Delta$ is small in the region of interest to construct a function that \textit{is} rapidly growing at $x=1$, but still has a non-singular interval of $x$ prior to $x=1$ for which a lower-bound on $f(x)$ does hold. This ``best of both worlds'' solution deploys the Dolph-Chebyshev window \cite{Dolph1946} (which was also used in a different, but related, application in Refs.~\cite{Gily2019, Grand-Unification}) as a singular value filter:

\begin{defn}[Dolph-Chebyshev window special case]
    Let 
    \begin{equation}
    f^{(3)}_d(x) = \frac{T_d \left( \frac{ 2x - (x_0 - 1)}{x_0 +1} \right) }{2 T_d \left( \frac{ 2 - (x_0 - 1)}{x_0 +1} \right)} 
\end{equation}
where 
\begin{equation}
    T_d(z) = \begin{cases} \cos( d \arccos z) & |z| \leq 1 \\ \cosh( d \, \textnormal{arccosh} \, z) & z \geq 1 \\
    (-1)^d \cosh( d \, \textnormal{arccosh} \, (-z)) & z \leq -1
    \end{cases}
\end{equation}
which is known as the Chebyshev polynomial of the first kind \cite{Rivlin1990}.
\end{defn}
\begin{remark}
    $f^{(3)}_d(x)$ is a polynomial in $x$ of degree $d$.
\end{remark}

To prove $f^{(3)}$ has the desired behaviour, we first need to establish a simple result. The exponential growth of the Chebyshev polynomial when the argument is close to one is well-known, however it does not (as far as we can see) appear in the literature in the exact form needed. 

\begin{lemma}
\label{lem:app}
For $r > 0$
    \begin{equation}
        T_d \left( 1 + \frac{1}{r} \right) = \frac{1}{2} e^{d (\sqrt{2 / r} + \mathcal{O}(1/r))} \left( 1 + o(1) \right)
    \end{equation}
\end{lemma}
\begin{proof}
By definition, $T_d \left( 1 + \frac{1}{r} \right) = \cosh \left( d \, \text{arccosh}  \, \left( 1 + \frac{1}{r} \right) \right)$, and first we focus on $\text{arccosh}  \, \left( 1 + \frac{1}{r} \right)$. By a the definition and a series of Taylor expansions,
\begin{align}
    \text{arccosh}  \, \left( 1 + \frac{1}{r} \right) & = \log_e \left( 1 + \frac{1}{r} + \sqrt{\left( 1 + \frac{1}{r} \right)^2 -1} \right) \nonumber \\
    & = \log_e \left( 1 + \frac{1}{r} + \sqrt{\frac{2}{r}} \sqrt{1 + \frac{1}{2r}} \right) \\ 
    & = \log_e \left( 1 + \frac{1}{r} + \sqrt{\frac{2}{r}} \left( 1 + \mathcal{O}\left( \frac{1}{r} \right) \right) \right) \\
    \label{eq:app-1}
    & = \sqrt{\frac{2}{r}} + \mathcal{O}\left( \frac{1}{r}  \right) 
\end{align}
We also have that by definition,
\begin{equation}
\label{eq:app-2}
    \cosh (d x) = \frac{ e^{d x} + e^{-d x}}{2} 
\end{equation}
The claim follows directly by substituting the right-hand side of (\ref{eq:app-1}) for $x$ in (\ref{eq:app-2}).
\end{proof}

We are now ready to show:

\begin{lemma}
    \label{lem:alg-filter-a}
For every fixed $0 < \alpha < \frac{1}{4}$, there exists $t^*, \, 0 < t^* < 1$ such that $\forall t \geq t^*$ and for $d \in \Theta \left( \frac{1}{\sqrt{\Delta}} \right)$, $x_0$ and $x_1$ can be set such that $-1 \leq x \leq 1, \, |f_d^{(3)}(x)| \leq 1$ and
    \begin{align}
          0 \leq x \leq 1 - (1+t)\Delta : & \,\,\,\,\,\,\, |f_d^{(3)}| \leq \alpha \\
        1 - (1-t)\Delta  \leq x \leq 1 :  & \,\,\,\,\,\,\, \frac{1}{4} \leq f_d^{(3)}
    \end{align}
\end{lemma}

\begin{proof}
     Set $x_0 = 1 - (1+t)\Delta$ and $x_1 = 1 - (1-t)\Delta$, then $f_d^{(3)}$ has three properties when $d$, the degree, is as stated in the lemma:
\begin{enumerate}[(i)]
    \item $f_d^{(3)}(1) = \frac{1}{2}$.
    \item $f_d^{(3)}(x)$ is monotonically increasing in the region $1 - (1+t) \Delta \leq x \leq 1$.
    \item There exists there exists $t^*, \, 0 < t^* < 1$ such that $\forall t \geq t^*$, the following can simulataneously be met: 
    \begin{align}
          0 \leq x \leq 1 - (1+t)\Delta : & \,\,\,\,\,\,\, |f_d^{(3)}| < \alpha \\
        1 - (1-t)\Delta  \leq x \leq 1 :  & \,\,\,\,\,\,\, \frac{1}{4} \leq f_d^{(3)}
    \end{align}
\end{enumerate}
Together these three are enough to conclude, and the first claim can be checked trivially from the definition, whilst the second follows from a standard property of the Dolph-Chebyshev filter: the denominator is a positive constant, whilst the argument of the numerator is increasing and lower-bounded by one -- and it is known that $T_d$ is monotonic in arguments greater than one \cite[Section 18.3.4]{NIST:DLMF}. Thus, to complete the proof, it is enough to show that the final claim holds.\\

For the first condition in claim (iii), in the region $ 0 \leq x \leq 1 - (1+t)\Delta$, the magnitude of the argument of $T_d$ in the numerator is bounded above by one, and hence the numerator is itself bounded above by one. So it is enough to show that
\begin{equation}
    T_d \left( \frac{ 2 - (x_0 - 1)}{x_0 +1} \right) > \alpha
\end{equation}
To analyse the argument, it is helpful to define $D = \frac{1}{\Delta}$, such that we have:
\begin{equation}
    \frac{ 2 - (x_0 - 1)}{x_0 +1}  = \frac{ 2 - ( 1 - (1+t)\Delta - 1)}{ 1 - (1+t)\Delta +1} = \frac{ 1 + (1+t)\Delta /2 }{ 1 - (1+t)\Delta /2 } = 1 + \frac{1+t}{D} + \mathcal{O}\left(\frac{1}{D^2} \right) 
\end{equation}
We also have for $r > 0$, the approximation Chebyshev polynomial established in Lemma~\ref{lem:app},
\begin{equation}
    T_d\left( 1+ \frac{1}{r} \right) = \frac{1}{2} e^{d (\sqrt{2 / r}+ \mathcal{O}(1/r)) } (1 + o (1 ) )
\end{equation}
Plugging in $\frac{1}{r} = \frac{1 + t}{D} + \mathcal{O}\left(\frac{1}{D^2} \right)$, we can see that for \textit{any} (fixed) $t$ and for any $k>0$, there exists $\Delta^*$ such that for any $\Delta < \Delta*$, then it is enough to choose any $n$ satisfying
\begin{equation}
\label{cond1}
    d > \frac{1+k}{\sqrt{2(1+t) \Delta}} \log \frac{2}{\alpha}
\end{equation}

Now turning to the second requirement, we know that the function is monotonically increasing for all $ x \geq 1 - (1-t) \Delta$, and so it is enough to consider the value of $f_d^{(3)}$ at $x = 1 - (1-t) \Delta$,
\begin{equation}
    f_d^{(3)}(x_1) = \frac{T_d \left( \frac{ 2x_1 - (x_0 - 1)}{x_0 +1} \right) }{2T_d \left( \frac{ 2 - (x_0 - 1)}{x_0 +1} \right)}  \geq \frac{1}{4}
\end{equation}
where we have already addressed the denominator, so turning to the argument of the numerator,
\begin{equation}
     \frac{ 2x_1 - (x_0 - 1)}{x_0 +1} = 1 + \frac{4 t \Delta}{2 - (1+t) \Delta} = 1 + \frac{2t}{D} + \mathcal{O} \left( \left( \frac{1}{D} \right)^2 \right)
\end{equation}
Putting this together and using the same expansion of the Chebyshev polynomial, for \textit{any} $t$ there exists $\Delta^*$ such that it is enough that for any $\tilde{k} > 0$, 
\begin{align}
    \frac{ 1 - \tilde{k}}{2} e^{d \sqrt{4 t \Delta} - d \sqrt{2(1+t) \Delta}} & \geq \frac{1}{4} \\
    \implies d \sqrt{\Delta} ( \sqrt{2 (1+t)} - 2 \sqrt{t} ) & \geq \log \frac{1}{2(1 - \tilde{k})} \\ 
    \label{cond2}
    \implies d \leq \frac{ \log (2(1 - \tilde{k}))}{ \sqrt{ \Delta}(\sqrt{2 (1+t)} - \sqrt{4t})}
\end{align}
We now turn to the conditions necessary for both (\ref{cond1}) and (\ref{cond2}) to hold, i.e., 
\begin{equation}
   \frac{1+k}{\sqrt{2(1+t) \Delta}} \log \frac{2}{\alpha}  < \frac{ \log (2(1 - \tilde{k}))}{ \sqrt{ \Delta}(\sqrt{2 (1+t)} - \sqrt{4t})}
\end{equation}
which can be rearranged to:
\begin{equation}
\label{eq:t*-1}
    t > t^* = \frac{C^2}{1 - C^2}
\end{equation}
where 
\begin{equation}
\label{eq:t*-2}
    C = \frac{1}{\sqrt{2}} - \frac{\log (2 (1- \tilde{k}))}{(1+k) \sqrt{2} \log \frac{2}{\alpha}}
\end{equation}
Which has the expected behaviour that $t^* \uparrow 1$ as $\alpha \to 0$. Thus for any fixed $\alpha$, as claimed we can find $t^*$, and guarantee the desired behaviour, and moreover by the above $d \in \Theta \left( \frac{1}{\sqrt{\Delta}} \right)$. 
\end{proof}

\begin{figure}[t!]
    \centering
    
    \begin{subfigure}{0.32\textwidth}
        \centering
        \includegraphics[width=\linewidth]{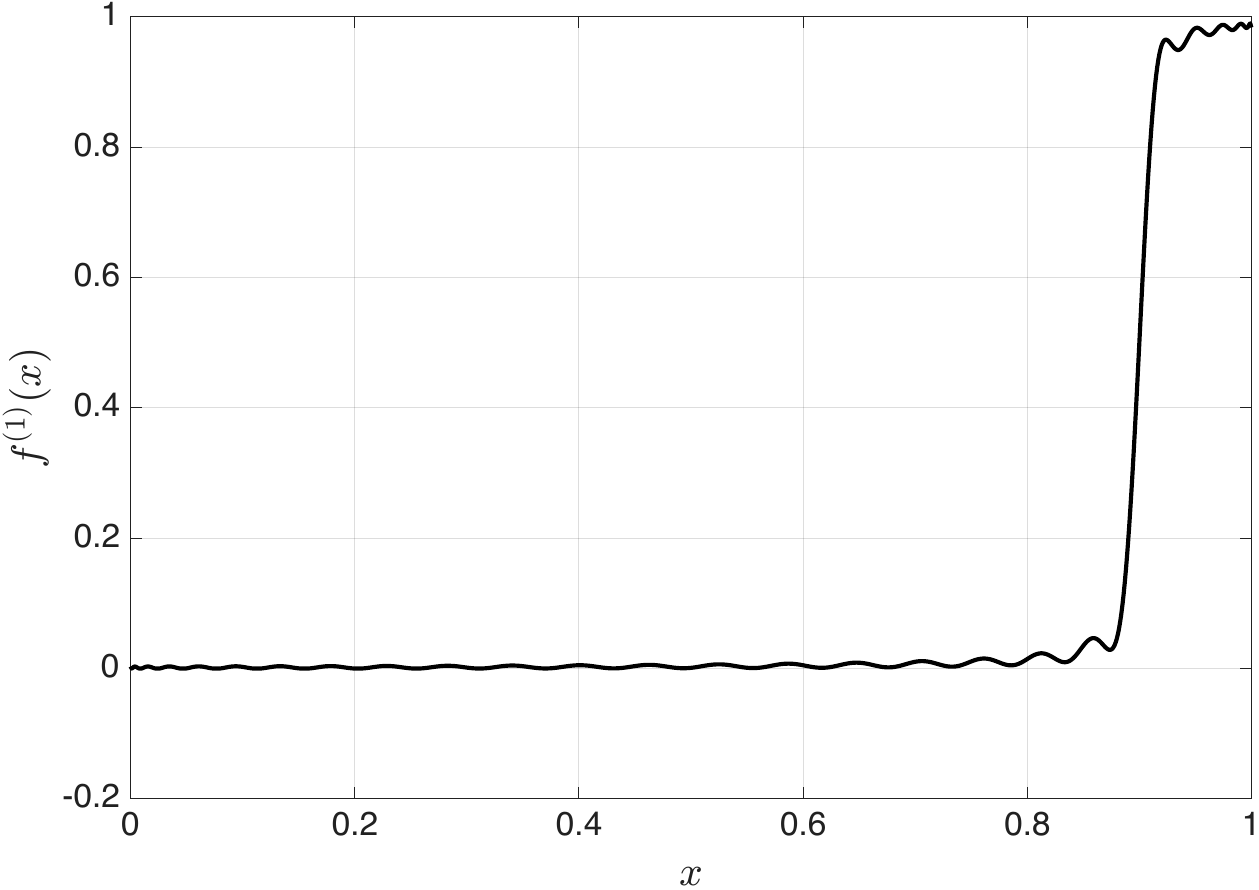}
        \caption{$f^{(1)}(x)$}
    \end{subfigure}
    \hfill
    \begin{subfigure}{0.32\textwidth}
        \centering
        \includegraphics[width=\linewidth]{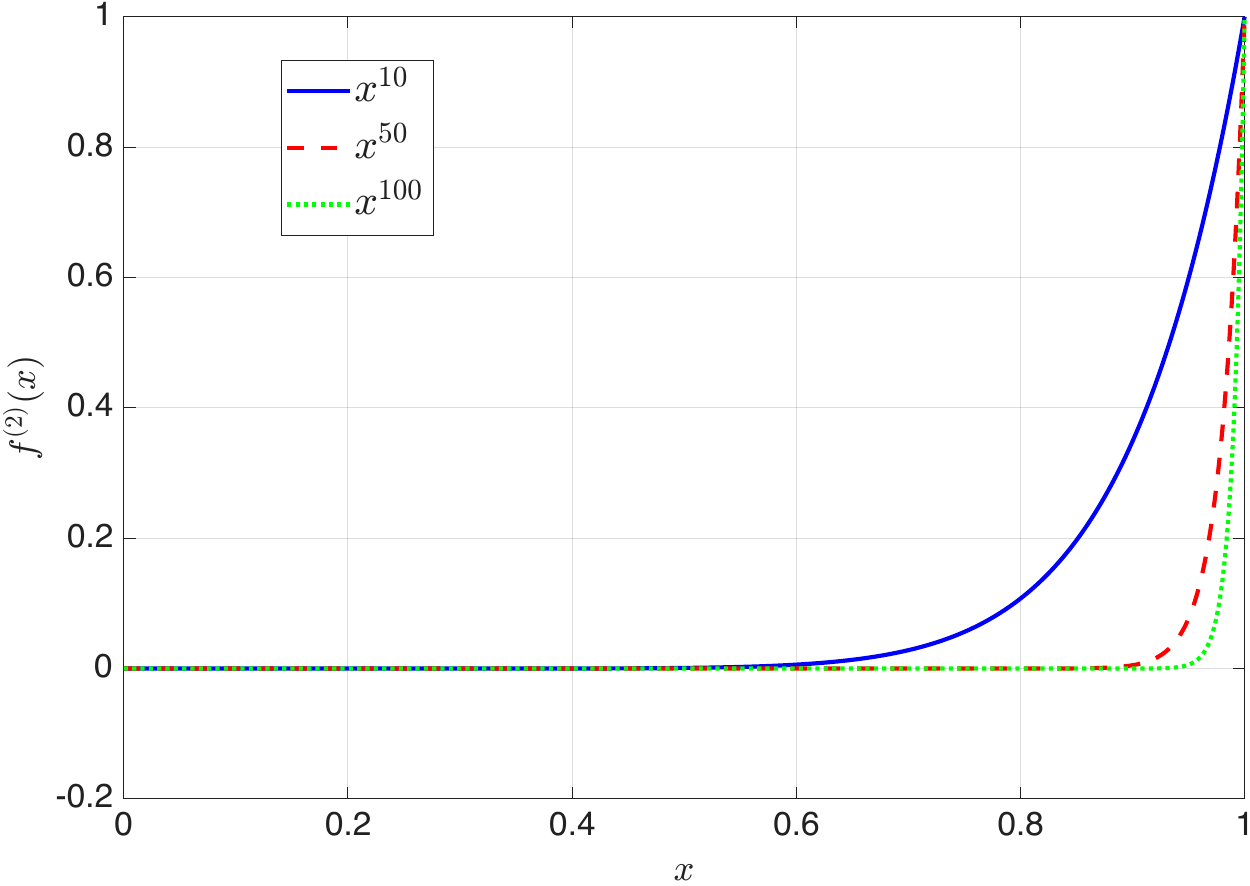}
        \caption{$f^{(2)}(x)$}
    \end{subfigure}
    \hfill
    \begin{subfigure}{0.32\textwidth}
        \centering
        \includegraphics[width=\linewidth]{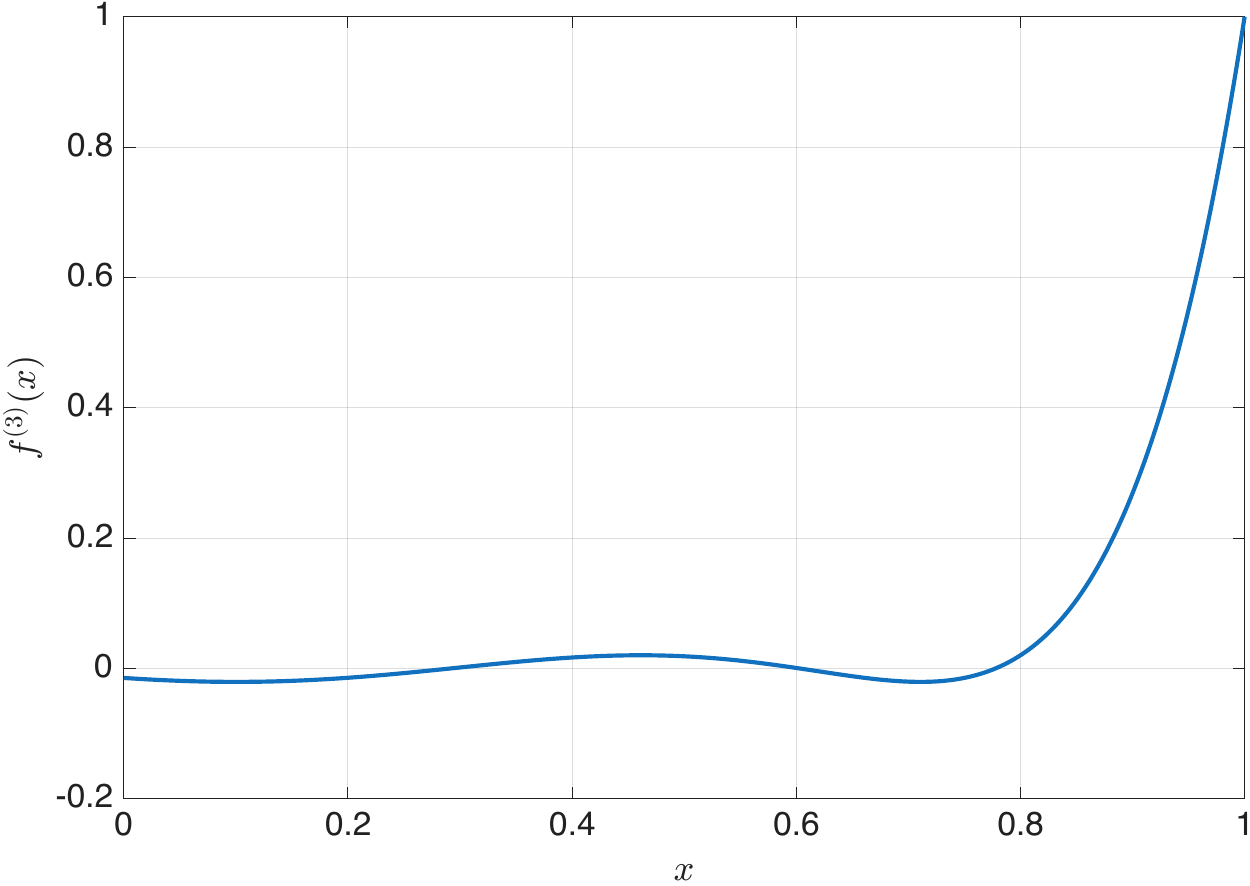}
        \caption{$f^{(3)}(x)$}
    \end{subfigure}

    \caption{A comparison of the three polynomial filters. In (a), the step function approximation is limited in how fast it can rise as it must ``turn flat again''; conversely there is no such limitation in (b), where the monomial (approximated by a polynomial with degree about square root of the monomial degree) grows rapidly at $x=1$. In the case of (c), the function grows slowly enough that there is a region of $x$ with a useful lower-bound (the attractive feature of (a)), whilst still requiring relatively low degree (the attractive feature of (b)).}
    \label{fig:chebyfig}
\end{figure}

The way that the three proposed functions, $f^{(1)}$, $f^{(2)}$ and $f^{(3)}$ achieve singular value filtering is sketched out in Fig.~\ref{fig:chebyfig}.

It is worth noting that $f^{(3)}$ is neither purely odd nor purely even, however every real function can be decomposed as a sum of an odd part plus an even part (with degree of each no greater than that of the original polynomial). Moreover, as $|f^{(3)}(x)| \leq \frac{1}{2}$ for $x \in [-1,1]$, the magnitudes of each of these parts is itself upper-bounded by 1 in the same region of $x$. Therefore each can be individually implemented using the QSVT, and the two parts then recombined using a linear combination of unitaries (LCU), as sketched in Ref.~\cite[Section VI-A]{Grand-Unification}. An elegant way of thinking about LCU for this purpose is an extra wrapping of block-encoding, however for the analysis it is also convenient to think of the LCU as preparing a post-selected state, which therefore fails with a certain probability in practice. In this case, there are two terms being linearly combined, and so by simple analysis the maximum failure probability is $\frac{1}{2}$. This failure is subsumed by another failure mode in the final algorithm, as briefly mentioned below. An alternative would be to use \textit{generalised quantum signal processing} \cite{motlagh2024generalized} to bypass the parity requirement, although this would require the block-encoding to be in a slightly different form to that which we use here.

\subsubsection*{Quantum counting}

The final component is to ``weigh'' the encoded block once the singular value filter has been applied, to decide if the second largest singular vector has been filtered out or not.

\begin{lemma}[Approximate quantum counting]
\label{lem-qcounting}
    Given an $n$-qubit quantum state:
    \begin{equation}
    \label{eq-lem-qcounting}
        \ket{\psi} = a \ket{0} \ket{\Psi_0} + \sqrt{\frac{a^2-1}{a^2}} \ket{1} \ket{\Psi_1}
    \end{equation}
    where $a$ is real and positive, which is prepared by a circuit, $A$, applied to the state $\ket{0^n}$; then defining $N=2^n$ the cases $a < c_l / \sqrt{N}$ and $a \geq c_u/\sqrt{N}$ for \textbf{fixed} real positive numbers $c_l < c_u$ can be distinguished (in the sense that $0$ should be output if $a < c_l/\sqrt{N}$ and $1$ output if $a \geq c_u/\sqrt{N}$ (for $c_l/\sqrt{N} \leq a < c_u/\sqrt{N}$, either 0 or 1 may be output); with quantum circuitry requiring $\Theta\left(\sqrt{N}\right)$ uses of $A$.
\end{lemma}
\noindent This is a straightforward application of quantum counting \cite{Qcount1, Qcount2, Qcount3}. It is convenient to define this as a formal subroutine that can be called:

\begin{defn}
    As proven to exist in Lemma~\ref{lem-qcounting}, let \textnormal{\texttt{QCount}}$(\ket{\psi}, c_l, c_u)$ be an algorithm that decides if the weight of an encoded block (i.e., the value of $a^2$ in (\ref{eq-lem-qcounting})) is at most $c_l^2 / N$ (output 0) or at least $c_u^2 / N$ (output 1). For other values either 0 or 1 may be output.
\end{defn}

\subsubsection*{Second largest eigenvalue thresholding}

With the three components, above, established, we return to the headline claim of this section. In particular, using \texttt{QCount} we propose Algorithm~\ref{alg-eig-disc}, which in turn gives the main result of this subsection. Note that the following result relies on the fact that the block-encoded matrix has maximum singular value equal to one -- which always holds for the symmetrised discriminant, and holds for unscaled block-encoding $P$ (as such $P$ is necessarily doubly-stochastic, as in Lemma~\ref{lem:when-block-encodable}).

\begin{algorithm}[!t]
\caption{Second largest singular value thresholding}\label{alg-eig-disc}
\begin{algorithmic}[1]
\Require $\tilde{U}_A$, $c_l$ and $c_u$ such that $c_l < c_u < 1$, $I$
\State $T = 0$
\For{$i = 1:I$}
\State Generate a $n$-qubit state, $\ket{\phi} = U\ket{0^n}$ from where $U \sim \mathcal{U}$.
\State $\ket{\psi} = \tilde{U}_A (\ket{0} \ket{\phi})$
\If{$\mathtt{QCount}(\ket{\psi}, c_l, c_u) =1$}
\State $T \gets 1$
\EndIf
\EndFor
\State \Return $T$
\end{algorithmic}
\end{algorithm}

\begin{thm}
\label{thm:singular-threshold}
    Let $U_A$ be a unitary block-encoding where $A$ is $N$-dimensional and either a Markov chain symmetrised discriminant, or the transition matrix of a doubly-stochastic Markov chain. There exist $\tilde{U}_A$, $c_l$, $c_u$ and $I$ which allow Algorithm~\ref{alg-eig-disc} to achieve the following for any $0 < L < 1$ and $\epsilon \leq \min (L, 1 - L)$ with failure probability any $0 < p_f$:
    \begin{itemize}
        \item outputs 0 if the second largest singular value of $A$ is smaller than $1 - L(1 + \epsilon)$;
        \item outputs 1 if the second largest singular value of $A$ is greater than $1 - L(1 - \epsilon)$;
    \end{itemize}
    Moreover, the overall number of uses of $U_A$ is:
    \begin{enumerate}[(i)]
    \item $\mathcal{O} \left(  \frac{\sqrt{N}}{\epsilon L } \log N  \log \frac{1}{p_f} \right) $, for arbitrary $\epsilon$.
    \item $\mathcal{O} \left(  \sqrt{\frac{N}{L} } \log N  \log \frac{1}{p_f} \right) $, for $\epsilon > \epsilon^*$ where $\epsilon^* > 0$ is a critical value that depends on $N$.
    \end{enumerate}
\end{thm}
\begin{proof}
    Let $\{\Ket{u_i} \}_{i=1}^N$ be the right singular vectors of $A$ expressed as quantum states, and such that $\sigma_i$ is the singular value corresponding to $\ket{u_i}$ which are ordered $\sigma_1 = 1 \geq \sigma_2 \geq \dots \geq \sigma_N \geq 0$. Further let $\tilde{U}_A$ be $U_A$ passed through the singular value filter of Lemma~\ref{lem:alg-filter-sign} or \ref{lem:alg-filter-a} (as appropriate) using the QSVT such that $\alpha = \frac{\tilde{c}_l}{N}$ for some constant $\tilde{c}_l < \frac{1}{4}$; also $\Delta = L$ and $t = \epsilon$. This guarantees that if every singular value apart from $\sigma_1$ is at most $1 - (L+\epsilon)$ then the weight of the encoded block is at most 
    \begin{equation}
        N (\tilde{c}_l / N)^2 + f(1)^2 |\braket{u_1 | \psi }|^2 = \frac{\tilde{c}_l^2 + f(1)^2 + \epsilon_{\psi}}{N}
    \end{equation}
    where $\epsilon_{\psi}$ captures the fact that $|\braket{u_1 | \psi }|^2$ is not exactly equal to $\frac{1}{N}$.

    On the other hand, if the second largest singular value has not been filtered out, then the weight of the block is at least 
    \begin{equation}
        (N-2) (\tilde{c}_l / N)^2 + \left(\frac{1}{4} \right)^2 |\braket{u_2 | \psi }|^2 + f(1)^2 |\braket{u_1 | \psi }|^2 = \frac{\tilde{c}_l^2 + f(1)^2 + \left(\frac{1}{4} \right)^2 + \epsilon'_{\psi}}{N} - 2 \frac{\tilde{c}_l^2 }{N^2}
    \end{equation}
    Noting that $f(1)$ is a fixed, known value that is itself lower-bounded by $\alpha$ then as per Remark~\ref{rem:two-design}, for sufficiently large $N$, and for any probability $p$, $0 < p < 1$, the values of $\epsilon_\psi$ and $\epsilon'_\psi$ are such that these two possibilities are non-overlapping with probability $p$. That is, $\left(\frac{1}{4} \right)^2 + \epsilon'_\psi$ exceeds $\epsilon_\psi$ by an amount that is independent of $N$ and lower-bounded by some constant. Thus quantum counting, Lemma~\ref{lem-qcounting}, can distinguish therebetween.

    Supposing that the sampled $\ket{\phi}$ indeed satisfies this, then the number of uses of $U_A$ is proportional to the degree of the polynomial transformation used in QSVT times the number of uses needed for quantum counting, that is, in the general case, when the filter of Lemma~\ref{lem:alg-filter-sign} is used
    \begin{equation}
    \label{eq:101}
        \Theta \left( \frac{1}{\Delta t } \log \frac{1}{\alpha} \right) \times \Theta(\sqrt{N}) = \Theta \left( \frac{\sqrt{N}}{\epsilon  L} \log \frac{1}{\tilde{c}_l / N} \right) = \Theta \left( \frac{\sqrt{N}}{\epsilon L} \log N \right) 
    \end{equation}
    Finally, we address the repeats. In Line 3, $\mathcal{O}(1)$ repeats are expected for the post-selection to succeed (note that the successful LCU preparation can also be subsumed into this statement); and to achieve failure $p_f$ it is enough to set $I = \mathcal{O}\left(\log \frac{1}{p_f} \right)$, thus giving the overall claimed complexity when multiplied in, satisfying claim (i).

    Suppose instead that for the given value of $N$, $\alpha = \frac{\tilde{c}_l}{N}$ yields a value of $t^*$ (as defined in (\ref{eq:t*-1}) and (\ref{eq:t*-2})) such that $\epsilon > \epsilon^* = t^*$, then $f^{(3)}$ can instead by deployed. In this case, the complexity of $f^{(3)}$ can be swapped in for that of $f^{(1)}$ in (\ref{eq:101}) to obtain the complexity of claim (ii).
\end{proof}

\subsection{Spectral and singular gap estimation}

It is relatively straightforward to use the algorithm of the previous subsection (i.e., to threshold the second largest eigenvalue) to estimate the spectral gap, but with a couple of small complications: (i) the method we propose is based on bisection search, but as the thresholding has a region of transition, the possible region after each application needs to be adjusted accordingly; (ii) it is necessary to be a little careful about the failure probability, such that the entire algorithm fails only with a certain probability, even though the number of iterations of the bisection search is unknown \textit{a priori}. To do this, it is convenient to define the algorithm of Theorem~\ref{thm:singular-threshold} explicitly:
\begin{defn}
    Let $\mathtt{SingularThreshold}(L, \epsilon, p_f)$ be the algorithm described in Theorem~\ref{thm:singular-threshold}.
\end{defn}
Which then gives Algorithm~\ref{alg:spec-gap}, whose performance is then proven in Theorem~\ref{thm:main}. For this it is important to be clear about the definition of relative error.

\begin{defn}
    Say an estimator, $\hat{x}$ of some real, positive quantity $x$ has relative error $\epsilon'$ if
    \begin{equation}
        (1 - \epsilon') x \leq \hat{x} \leq x (1 + \epsilon') 
    \end{equation}
\end{defn}
For our purposes it is convenient to define the closely related notion:
\begin{defn}
    Say an estimator, $\hat{x}$ of some real, positive quantity $x$ has relative' error $\epsilon$ if
    \begin{equation}
        (1 - \epsilon) \hat{x} \leq x \leq \hat{x} (1 + \epsilon) 
    \end{equation}
\end{defn}
It is easy to see that these definitions match up to second order terms in $\epsilon$ (or $\epsilon'$) and so for small relative error essentially coincide. Moreover, for any fixed relative error there is a direct conversion between the two. Hence for simplicity of presentation we press on using relative' error for the main result.

\begin{algorithm}[!t]
\caption{Singular gap estimation to relative' error $\epsilon_\gamma$ with failure probability $p_f$}\label{alg:spec-gap}
\begin{algorithmic}[1]
\Require Unitary block encoding of the transition matrix or symmetrised discriminant, $U_A$; $\epsilon_\gamma$, $p_f$.
\State $\gamma_{\text{max}} = 1$; $\gamma_{\text{min}} = 0$; $p_f' = p_f$.
\While{$\hat{\gamma}(1 + \epsilon_\gamma) < \gamma_\text{max}$ OR $\hat{\gamma}(1 - \epsilon_\gamma) > \gamma_\text{min}$}
\State $\hat{\gamma} = \frac{1}{2}(\gamma_{\text{max}} + \gamma_{\text{min}})$
\State $\epsilon \gets \frac{1}{4} (\gamma_{\text{max}} - \gamma_{\text{min}})$
\State $p_f' \gets \frac{1}{2} p_f'$
\State $T = \mathtt{SingularThreshold}(\hat{\gamma}, \epsilon, p_f)$
\If{$T = 1$}
\State $\hat{\gamma}_{\text{max}} \gets \hat{\gamma} + \epsilon$
\Else{}
\State $\hat{\gamma}_{\text{min}} \gets \hat{\gamma} - \epsilon$
\EndIf
\EndWhile
\State Return $\hat{\gamma}$
\end{algorithmic}
\end{algorithm}

\begin{thm}
\label{thm:main}
    Using Alorithm~\ref{alg:spec-gap}, the singular gap, $\gamma$, of a Markov chain symmetrised discriminant or transition matrix if doubly-stochastic, can be estimated to relative' error $\epsilon_\gamma \leq 1$ and failure probability $p_f$, with 
    \begin{enumerate}[(i)]
    \item $\tilde{\mathcal{O}}  \left( \frac{\sqrt{N}}{ \epsilon_\gamma \gamma} \log \frac{1}{p_f}  \right) $ uses of a block-encoding of the symmetrised discriminant or transition matrix, for arbitrary $\epsilon_\gamma$.
    \item $\tilde{\mathcal{O}}  \left( \frac{\sqrt{N}}{  \sqrt{\gamma}} \log \frac{1}{p_f}  \right) $  uses of a block-encoding of the symmetrised discriminant or transition matrix, for $\epsilon_\gamma > \epsilon_\gamma^*$ where $\epsilon_\gamma^* > 0$ is a critical value that depends on $N$.
    \end{enumerate}
\end{thm}

\begin{proof}

First note that the algorithm is based on bisection search and clearly has the correct function when it does not ``fail''; moreover, by setting $p'_f \gets \frac{1}{2} p'_f$ as in Line~5 at each iteration, the union bound guarantees a total failure probability less than $p_f$. So it remains to count the complexity of: (i) \texttt{SingularThreshold} in Line 6 for each call; and (ii) the number of iterations of the while loop. 

Taking these in turn and for now assuming that no failure occurs, it is enough to consider the worst case for each parameter of any call. For $\hat{\gamma}$, we know that $\gamma$ lies in the region $[\gamma_\text{min}, \gamma_\text{max}]$ at each iteration and, as such, the minimal value of $\hat{\gamma}$ for which \texttt{EigenThreshold} will be called is 
\begin{equation}
\hat{\gamma} \geq \frac{1}{2}(\gamma_{\text{max}} + \gamma_{\text{min}}) \geq \frac{\gamma_{\text{max}}}{2} \geq \frac{\gamma}{2}
\end{equation}
Turning to $\epsilon$, in the worst case at the penultimate iteration $\gamma_\text{max} - \gamma_\text{min}$ is marginally larger than $2\epsilon_\gamma$, and so a further iteration is required with $\epsilon = \frac{1}{4} (\gamma_{\text{max}} - \gamma_{\text{min}}) > \frac{1}{2} \epsilon_\gamma$. Finally, the worst case of $p_f'$ is $p_f 2^{- I}$, where $I$ is the number of iterations of the while loop, which is item (ii), to which we now turn.

For this, we note that at each iteration the region $[\gamma_\text{min}, \gamma_\text{max}]$ shrinks by a factor $\frac{3}{4}$, and  so the smallest $I$ that satisfies 
\begin{equation}
    \left( \frac{3}{4} \right)^I < 2 \epsilon_\gamma
\end{equation}
and so we have $I \in \Theta \left( \log \frac{1}{\epsilon_\gamma} \right)$; which also gives $\min p_f' \in \Theta (p_f \epsilon_\gamma)$

Putting this all together, for the worst case parameters and $I$ iterations of the while loop gives:
\begin{equation}
    \Theta \left( \frac{\sqrt{N}}{\epsilon_\gamma \gamma} \log N \log \frac{1}{p_f \epsilon_\gamma} \log \frac{1}{\epsilon_\gamma} \right) = \tilde{O} \left( \frac{\sqrt{N}}{\epsilon_\gamma \gamma} \log \frac{1}{p_f }  \right)
\end{equation}
when \texttt{SingularThreshold} uses $f^{(1)}$, as in claim (i). Conversely, if the value of $N$ and $\epsilon_\gamma$ are such that $f^{(3)}$ can be deployed, then the complexity of claim (ii) follows. (A very minor subtlety being that as this only works for $\Delta < \Delta^*$, for the early iterations of the while loop, $f^{(1)}$ will still be deployed, but it is guaranteed that $f^{(3)}$ will be switched to at some fixed $\Delta^*$, hence the overall complexity is as claimed.)
\end{proof}

It is important to note that this upper-bound on the number of uses of the block-encoding holds when the algorithm does not fail. Suppose that on a certain iteration the eigenvalue threshold does fail, which will result in the region $[1 - \gamma_{\max}, 1 - \gamma_{\min}]$ not actually containing a singular value. Supposing the rest of the algorithm runs correctly, then on each occasion the reduced region will constitute that corresponding to the largest spectral gap, and as such that final result (though incorrect) will be as large as possible (given the failure) and so will have large relative error, and so will terminate promptly. Thus, though the algorithm may with low probability fail to return the correct value, even on these occurrences the run-time decays exponentially, as a catalogue of failures would be required to keep it running.

\subsection{Quasi-optimality of the spectral gap estimation algorithm}

We now show that our algorithm is quasi-optimal in each of $N$, the number of vertices, and $\Delta$ the spectral / singular gap.

\begin{proposition}
    Any algorithm to decide if the singular gap of is greater or less than some $\Delta$ of any Markov chain symmetrised discriminant or transition matrix when the maximum singular value equals one necessarily requires $\Omega (\sqrt{N})$ uses of the Markov chain symmetrised discriminant / transition matrix block encoding.
\end{proposition}
\begin{proof}
This is proved by showing that faster scaling would violate the $\Omega(\sqrt{N})$ bound on unstructured search. In particular, based on Ref.~\cite{Szegedy}, the unstructured search oracle can be used to construct a unitary (that can be viewed as a) block-encoding of $D$ for two different cases:
\begin{enumerate}
    \item When no element is marked, then the Markov chain amounts to random sampling, i.e., 
    \begin{equation}
        P = \frac{1}{N} \begin{bmatrix} 1 & 1 & \dots & 1 \\
        1&  1 & \dots & 1 \\ \vdots & \vdots & \ddots & \\
        1 & 1 & & 1 \end{bmatrix} \implies D = \frac{1}{N} \begin{bmatrix} 1 & 1 & \dots & 1 \\
        1&  1 & \dots & 1 \\ \vdots & \vdots & \ddots & \\
        1 & 1 & & 1 \end{bmatrix}
    \end{equation}
    Which has singular values $1$ (unique) and 0 (with multiplicity $N-1$). 
    \item When there is a single marked element (shown as the $N$th element for illustrative simplicity, but this is not required for the following analysis, which is completely general), then
    \begin{equation}
        P = \frac{1}{N} \begin{bmatrix} 1 & 1 & \dots & 1 \\
        1&  1 & \dots & 1 \\ \vdots & \vdots & \ddots & \\
        0 & 0 & & N \end{bmatrix} \implies D = \frac{1}{N} \begin{bmatrix} 1 & 1 & \dots & 0 \\
        1&  1 & \dots & 0 \\ \vdots & \vdots & \ddots & \\
        0 & 0 & & N \end{bmatrix}
    \end{equation}
    Which has singular values $1$ (unique), $\frac{N-1}{N}$ (unique) and $0$ (with multiplicity $N-2$).
\end{enumerate}
So it follows that, for example, running \texttt{SingularThreshold}$\left( \frac{1}{2}, \frac{1}{2}, \frac{1}{3} \right)$ distinguishes these two possibilities for all $N$ greater than $4$. So should this do so with scaling better than $\Theta (\sqrt{N})$ then it would be a violation of the known lower-bound on the complexity of deciding the OR problem \cite{quadratic-lower-bound}.
\end{proof}

The lower-bound on deciding the OR problem is typically proven using the adversary method \cite{Ambainis2002,Ambainis2009} or the polynomial method \cite{BBCMW2001}. The latter provides a template to show that the version of the algorithm deploying $f^{(3)}$ is additionally quasi-optimal in terms of the scaling with $\frac{!}{\Delta}$.

\begin{proposition}
    No algorithm (whose initial state is independent of the singular values) to decide if the singular gap is greater or less than some $\Delta$ to some fixed relative error can do so with complexity scaling less than $\Omega \left( \frac{1}{\sqrt{\Delta}} \right)$.
\end{proposition}
\begin{proof}
This can be proved by an application of the polynomial method. Let the algorithm act on an initial state $\ket{0} \ket{0^n} \ket{0^m}$, where the block encoding always is applied such that it is the $\ket{0} \ket{0^n}$ register that is transformed. Moreover, let $\{ \ket{u_i} \}_{i=1}^N$ be the full set of right singular vectors of the block-encoded matrix $A$. As such any state of the correct size can be expressed:
\begin{equation}
    \ket{\psi} = \sum_i \sum_j a_{i,j} \ket{0} \ket{u_i} \ket{j} + \sum_i \sum_j b_{i,j} \ket{1} \ket{u_i} \ket{j}
\end{equation}
We now consider that any circuit applied to the state can be represented $\prod_{k=1}^d W_i (U_A \otimes I)$ where $W_i$ are unitary blocks that are independent of the singular values, and may include swap gates so that there is no loss of generality in always assuming the block-encoding $U_A$ is applied to the initial qubits. Furthermore, if $\ket{\psi^{(d)}} = \prod_{k=1}^d W_i (U_A \otimes I) \ket{0^{n+m+1}}$ is the final state then,
\begin{equation}
    \ket{\psi^{(d)}} = \sum_i \sum_j a_{i,j}^{(d)} \ket{0} \ket{u_i} \ket{j} + \sum_i \sum_j b_{i,j}^{(d)} \ket{1} \ket{u_i} \ket{j}
\end{equation}
then $a_{i,j}^{(d)}$ and $b_{i,j}^{(d)}$ are polynomials in $\{\sigma_i\}_{i=1}^N$ of degree at most $d$. Including in the final unitary block the operation needed to map the 0 or 1 output bit to a single qubit measurement, then we can further see that the measurement outcome probability is itself a polynomial in $\{\sigma_i\}_{i=1}^N$ of degree $2d$ (by the Born rule). 

If we let the relative error be some constant $t$ which is chosen such that the version of the algorithm using $f^{(3)}$ can be deployed, and (departing slightly from the above established convention) further consider the singular values $\{ \sigma_i \}_{i=1}^N$ to simply be indexed, but not necessarily ordered, then we can consider measurement outcome as a function, $F$, of $\{ \sigma_i \}_{i=1}^N$, where each singular value varies between 0 and 1. For concreteness we may consider a family of Markov chains with some $N$ vertices, where the $N^{\text{th}}$ state is absorbing, and every other vertex has a self-loop with probability $\lambda_i$ for the $i^{\text{th}}$ vertex and with probability $1 -  \lambda_i$ transitions to the absorbing state. This family of Markov chains has symmetrised discriminant that is diagonal in the computational basis, and apart from the absorbing state has arbitrary single values:
\begin{equation}
    D = \begin{bmatrix}
        \lambda_1 & 0 & \dots & 0 \\ 0 & \lambda_2 & \dots & 0 \\ 
        \vdots & \vdots & \ddots &  \\ 
        0 & 0 &  & 1 \\ 
    \end{bmatrix}
\end{equation}
In particular, in this case the aforementioned function $F$ is parameterised by $\{ \sigma_i \}_{i=1}^{N-1}$ as $\sigma_i = \lambda_i$ for $i = 1 \dots N-1$ (and $\sigma_N = 1$). If we consider a slice of this function corresponding to $\lambda_2 \dots \lambda_{N-1} = 0$, then the singular gap is determined only by the value of $\lambda_1$. For notational simplicitly, we call this slice $\tilde{F}$ -- a function of a single argument ($\lambda_1$) -- and to correctly decide the spectral gap, must satisfy:
\begin{align}
    \tilde{F}(\Delta(1+t)) & \leq \frac{1}{3} \\ 
    \tilde{F}(\Delta(1-t)) & \geq \frac{2}{3} 
\end{align}
This means that $\tilde{F}$ must rise by an amount $\frac{1}{3}$ in a width $2 \Delta t$, and so it must have a gradient at least $\frac{1 / 3}{2 \Delta t} = \frac{1}{6 \Delta t}$ at some point. Moreover, by the Markov brothers' inequality \cite{Markov1890}, 
\begin{equation}
    \Bigg| \frac{\mathrm{d} \tilde{F}(x)}{\mathrm{d} x} \Bigg| \leq (\text{deg} \tilde{F})^2 \frac{\tilde{F}_{\max} - \tilde{F}_{\text{min}}}{x_{\max}-x_{\min}}
\end{equation}
and as we are concerned with box $[0,1] \times [0,1]$ the final term equals one. Moreover, substituting in $(\text{deg} \tilde{F}) \leq 2d$ and $\Big| \frac{\mathrm{d} \tilde{F}(x)}{\mathrm{d} x} \Big| \geq \frac{1}{6 \Delta t}$, we get
\begin{equation}
    \frac{1}{6 \Delta t} \leq (2d)^2  \implies d \geq \sqrt{ \frac{1}{24 \Delta t}} \implies d \in \Omega \left( \frac{1}{\sqrt{\Delta}} \right)
\end{equation}
as claimed. 
\end{proof}

\section{Discussion}
\label{sec:disc}

This paper proposes a quantum algorithm for Markov chain estimation that is quasi-optimal in the number of vertices for all parameters, and additionally quasi-optimal in the reciprocal of the spectral gap itself if the relative error is allowed to be above a certain threshold. This improves on the state of the art in terms of asymptotic scaling, and, as motivated in the introduction, is a practically-useful quantum advantage stemming from the QSVT. We also give two explicit block-encoding methods for Markov chain transition matrices -- to our knowledge the first such block encodings when the transition matrix does not coincide with the Markov chain symmetrised discriminant. The results in this paper also suggest a number of future research directions, and in particular prompts the following three questions.

\subsection{Open problems}

\begin{enumerate}
    \item Can we find explicit block-encoding methods for other types of doubly-stochastic matrices.
    \item Can we improve upon the two filtering polynomials deployed in this paper. In particular, can we find a filter that retains the optimal scaling in $\frac{1}{\Delta}$ and works for a greater range of $t$ than does $f^{(3)}$.
    \item Can we estimate or decide other Markov chain / graph properties using the block-encoding techniques herein. For instance, to give a concrete example of where a Markov chain of the type considered in Section~\ref{subsec:group-P} might arise, suppose that we would like to determine if a graph is bipartite. If our graph contains $N$ vertices, we could define a Markov chain over the group $\{0,1\}^N$, where we identify $\{0,1\}$ with the cyclic group of order $2$. An element $(a_1,...,a_N)$ in $\{0,1\}^N$ represents the partition with the $i$th vertex of the graph contained in the part $a_i$ (for example, $(1,1,0,0)$ means vertices $1,2$ are in one part and $3,4$ are in another). We'd mark elements of $\{0,1\}^N$ if they satisfy the required constraints that no adjacent vertices lie in the same part. We are free to define $\mu$ however we'd like, and an important research question is whether this could chosen to obtain a quantum advantage. Another important research question is whether this idea can be generalized to $k$-partiteness by replacing $\{0,1\}$ with the cyclic group of order $k$.
\end{enumerate}

\subsection*{Acknowledgements}

Many thanks to Yuta Kikuchi and Greg Boyd for reviewing this paper.


\end{document}